\documentclass[11pt,reqno]{amsart}

\usepackage{amsmath, amsthm, amstext, amssymb, amsfonts, graphicx, bm, subfigure}
\renewcommand{\L}{\mathcal{L}}

\newtheorem{theorem}{Theorem}[section]

\newtheorem{definition}[theorem]{Definition}
\newtheorem{lemma}[theorem]{Lemma}
\newtheorem{proposition}[theorem]{Proposition}
\newtheorem*{remark}{Remark}

\begin{document}

\title{Renormalization Group Transformations Near the Critical Point: Some Rigorous Results}

\author{Mei Yin}
\address{Department of Mathematics, University of Texas, Austin, TX
78712, USA} \email{myin@math.utexas.edu}

\dedicatory{\rm October 19, 2011}

\begin{abstract}
We consider renormalization group (RG) transformations for
classical Ising-type lattice spin systems in the infinite-volume
limit. Formally, the RG maps a Hamiltonian $H$ into a renormalized
Hamiltonian $H'$:
\begin{equation*}
\exp(-H'(\sigma'))=\sum_{\sigma}T(\sigma,
\sigma')\exp(-H(\sigma)),
\end{equation*}
where $T(\sigma, \sigma')$ denotes a specific RG probability
kernel, $\sum_{\sigma'}T(\sigma, \sigma')=1$, for every
configuration $\sigma$. With the help of the Dobrushin uniqueness
condition and standard results on the polymer expansion, Haller
and Kennedy gave a sufficient condition for the existence of the
renormalized Hamiltonian in a neighborhood of the critical point.
By a more complicated but reasonably straightforward application
of the cluster expansion machinery, the present investigation
shows that their condition would further imply a band structure on
the matrix of partial derivatives of the renormalized interaction
with respect to the original interaction. This in turn gives an upper bound
for the RG linearization.
\end{abstract}

\maketitle

\section{Introduction}
\label{intro} We consider renormalization group (RG)
transformations of finite-range and translation-invariant
Hamiltonians. Among possible RG transformations, there may be some
that are good, in the sense that they have a non-trivial fixed
point with desirable properties. This fixed point would represent a
critical state and is invariant under the RG transformation. It
sits on a critical surface, which consists of all those distinct
Hamiltonians whose critical trajectories under the RG map converge
to it. For a point not on the critical surface but very close to
being critical, the RG map will first drive it towards the fixed
point for a large number of iterations, but eventually will drive
it away. In the critical region, the thermodynamic systems are
characterized by long-range correlations among
microscopic fluctuations of local quantities that persist out to
macroscopic wavelengths. There is the astonishing empirical fact
that certain exponents associated with critical phenomena are
universal, and in particular, they are related by the scaling laws
to eigenvalues of the linearized RG map near the fixed point
\cite{Ma}. We are therefore interested in studying the behavior of
the RG transformation when the Hamiltonian is on or near the critical
surface, and most preferably in a neighborhood of the fixed point.

The relevant mathematical work is extensive, but it is fair to say
that many if not most questions remain unanswered. One could worry
about the various issues raised by van Enter, Fern\'{a}ndez, and
Sokal \cite{vanEnter2}, questioning whether the transformation is
even defined. Fortunately, if one thinks of the expansions
defining the RG transformation as analogous to other expansions in
statistical mechanics, then even if the original system is at the
critical point, under certain conditions, the systems that must be
studied to define the RG map need not be critical. Therefore there
is hope of using convergent expansions to define and analyze the
properties of the RG transformation.

Aizenman \cite{Aizenman} proved and explained a number of basic
features of the critical behavior of Ising models and $\phi^4$
fields in high and low dimensions by a non-perturbative analysis
of the field variables. Rivasseau \cite{Rivasseau} used
perturbative and constructive renormalization to investigate
rigorously the phenomenon of asymptotic freedom. Martinelli and
Olivieri \cite{Mar1, Mar2} investigated the stability and
instability of pathologies of RG transformations under decimation.
Haller and Kennedy \cite{Haller} showed that a single RG
transformation could map an area including a critical point to a
set of well-defined renormalized interactions. (See Fisher \cite{Fisher} for a
brief introduction to the historical developments in RG theory.)

Triggered by the realization that a single-site stochastic RG map,
say Kadanoff transformation with time-dependent parameter $p(t)$,
could be viewed as an infinite-temperature Glauber dynamics,
further inquiries followed. Van Enter, Fern\'{a}ndez, den
Hollander, and Redig \cite{EF} studied the time evolution of a low
non-zero temperature Gibbs state of Ising spins under infinite
temperature Glauber dynamics and showed that the evolved state is
Gibbsian for short time always, but non-Gibbsian for long enough
times. Le Ny and Redig \cite{LR} proved that for a short interval
of time a Gibbs measure with a finite range interaction evolved
under a general local stochastic dynamics would always remain
Gibbsian. Maes and Neto\v{c}n\'{y} \cite{MN} considered classes of
both discrete time and continuous time interacting particle
systems in the weak coupling regime and identified sufficient
conditions for which the time-evolved measure is Gibbsian for all
(even infinite) times. K\"{u}lske and Opoku \cite{KO} extended the
notion of Gibbsianness for mean-field systems to the setup of
continuous local state spaces and generalized previous case
studies made for spins taking finitely many values. (See \cite{EK}
for a review on some recent developments in the study of Gibbs and
non-Gibbs properties of transformed $n$-vector lattice and
mean-field models under various transformations). More references
to rigorous results inspired by, or implementing RG ideas, may be
found in Brydges \cite{Brydges}, Faris \cite{Faris}, Feldman et
al. \cite{Feldman}, and Mitter \cite{Mitter}.

Formally, the RG maps a Hamiltonian $H(\sigma)=-\sum_X
J(X)\sigma_X$ into a renormalized Hamiltonian $H'(\sigma')=-\sum_Y
J'(Y)\sigma'_Y$:
\begin{equation}
\exp(-H'(\sigma'))=\sum_{\sigma}T(\sigma,
\sigma')\exp(-H(\sigma)),
\end{equation}
where $J$ is the original interaction, $J'$ is the renormalized
interaction, and $T(\sigma, \sigma')$ is a probability kernel,
$\sum_{\sigma'}T(\sigma, \sigma')=1$, for every configuration
$\sigma$. Our basic assumption is that the original interaction
$J$ lies in a Banach space $\mathcal{B}$, with norm
\begin{equation}
||J||=\sum_{X \ni \hspace{0.05cm} 0}|J(X)|.
\end{equation}
Since the interaction space $\mathcal{B}$ does not put any
additional restrictions on the interactions other than absolute
summability, it is the largest physically reasonable space of
interactions. However, as argued in \cite{vanEnter2}, it may be
too large for a useful implementation of RG ideas near a critical
point as some rather strange phenomena would occur. To study the
behavior of the RG transformation when the system is at
criticality, extra assumptions are needed. A key tenet of the
renormalization group is that the introduction of the block spins
would shift the location of the critical point, thus making the
critical situation treatable by analytic methods. Haller and
Kennedy \cite{Haller}, using the Dobrushin-Shlosman machinery,
adopted this approach. They defined a probability measure
$\mu_{\sigma', V, \tau}$ which depends on the finite volume $V$,
the boundary condition $\tau$ and the block spin configuration
$\sigma'$ by
\begin{equation}
\mu_{\sigma', V, \tau}(F)=\frac{\sum_{\sigma}F(\sigma)T(\sigma,
\sigma')e^{-H(\sigma)}}{\sum_{\sigma}T(\sigma,
\sigma')e^{-H(\sigma)}},
\end{equation}
where $F(\sigma)$ is a function on the original spin configuration
$\sigma$. With the help of the Dobrushin uniqueness condition and
standard results on the polymer expansion, they gave a condition
which is sufficient to imply that the renormalized Hamiltonian
$H'$ is defined. Roughly speaking, the condition is that the
collection of measures $\mu_{\sigma'}$ is in the high-temperature
phase uniformly in the block spin configuration $\sigma'$.

\textbf{Hypothesis}: There exist constants $c<\infty$ and $m>0$
such that for every finite subset $V$ of the lattice, every two
sites $i, j \in V$, every boundary condition $\tau$, and every
block spin configuration $\sigma'$,
\begin{equation}
|\mu_{\sigma', V, \tau}(\sigma_i \sigma_j)-\mu_{\sigma', V,
\tau}(\sigma_i)\mu_{\sigma', V, \tau}(\sigma_j)| \leq
ce^{-m|i-j|}.
\end{equation}

They verified this condition numerically in two special cases
\cite{Haller}: Decimation with spacing $b = 2$ on the square
lattice for $\beta<1.36\beta_c$, and the Kadanoff transformation
with parameter $p$ on the triangular lattice in a subset of the
$(\beta, p)$ plane that includes values of $\beta$ around
$\beta_c$. Although critical fixed points do not arise after
infinite iterations of the decimation transformation, and the
limiting behavior of the system under the Kadanoff transformation
still seems to be an open question, this problem does not show up
after a finite number of applications of these maps, so we will
not worry too much about it in our present investigation.
Furthermore, since this hypothesis is similar to the
Dobrushin-Shlosman complete analyticity condition, one would
expect that it holds in a more general setting, and provides a
reasonable assumption for inquiry into the behavior of the RG
transformation near the critical surface.

Using cluster expansion techniques as in \cite{Yin2}, we show that
the \textbf{Hypothesis} also guarantees the existence of the partial
derivatives of the RG transformation (Theorem \ref{expar}). By a
more careful analysis, we will then show that the partial
derivative decays sub-exponentially as the distance between the
set in the original lattice and the set in the image lattice gets
large. It follows that the matrix of partial derivatives displays
an approximate band property (Theorem \ref{band}). This in turn
gives an upper bound for the RG linearization (Theorem \ref{lin}).
These results extend my previous work on the behavior of the RG
transformation at infinite temperature \cite{Yin1}, which shows
that the RG spectrum corresponding to decimation and majority rule
is of an unusual kind: dense point spectrum for which the adjoint
operators have no point spectrum at all, but only residual
spectrum. Thus, although the RG transformation exists, its properties appear not at all to be
what one would expect from the physics literature predictions.

For notational convenience, we will denote
$T(\sigma,\sigma')e^{-H(\sigma)}$ by $e^{-H}$ in the following. As
shown in \cite{Haller}, this modified Hamiltonian $H$ is also
finite-range.

\begin{proposition}
\label{expand} For every subset $W$ of the original lattice and
every subset $Z$ of the image lattice, the partial derivative
$\frac{\partial J'(Z)}{\partial J(W)}$ of the RG transformation is
given by the expression
\begin{equation}
\label{part} \frac{\partial J'(Z)}{\partial
J(W)}=\sum_{\sigma'}\sigma'_Z\frac{\sum_{\sigma}\exp(-H)\sigma_W}{\sum_{\sigma}\exp(-H)}.
\end{equation}
\end{proposition}

\begin{proof}
The renormalized coupling constants $J'$ are given by
\begin{equation}
\label{JZ} J'(Z)=\sum_{\sigma'}\sigma'_Z
\log\left(\sum_{\sigma}T(\sigma,\sigma')e^{\sum_X
J(X)\sigma_X}\right).
\end{equation}
We take the derivative of both sides of (\ref{JZ}) with respect to
$J(W)$.
\end{proof}

\begin{definition}
For every subset $Z$ of the image lattice, the linearization
$\mathrm{L}(J_c)$ of the RG transformation at a critical point
$J=J_c$ is given by a linear function of the deviation K:
\begin{equation}
\label{linear} \mathrm{L}(J_c)K(Z)=\sum_W \left.\frac{\partial
J'(Z)}{\partial J(W)}\right|_{J=J_c}K(W),
\end{equation}
where $W$ ranges over all finite subsets of the original lattice.
\end{definition}

\begin{remark}
The above calculations are only rigorous for finite lattices, but
may be interpreted in some more sophisticated limiting sense for
infinite lattices, following standard interpretation of
statistical mechanics, as will be shown in later sections.
\end{remark}

Many important physical properties emerge from spectral properties
of the linearization of the RG map. For an interaction very close
to being critical ($J=J_c+K$ with $K$ small), its behavior under
the RG transformation will be governed by the linearization
$\mathrm{L}(J_c)$:
\begin{equation}
\label{last}
(J_c+K)'(Z)=J_c'(Z)+\mathrm{L}(J_c)K(Z)+\text{corrections}.
\end{equation}
The main difficulty in this approximation is that the original
interaction $J$ is not small near the critical point, thus a
direct cluster expansion is not applicable. However, there is a
marvellous estimate on long range energies that will provide us
with the smallness needed. We first review relevant results
\cite{Haller}, where the existence of the renormalized interaction
$J_c'$ was justified.

\section{Review of relevant results}
\label{sec:1} The key idea \cite{Haller} is to divide the original
lattice $\L$ into blocks that are $L$ sites long on each side.
Haller and Kennedy referred to them as $L$-blocks (indexed by
$\bar{\L}$), and chose $L$ large enough so that these $L$-blocks
are commensurate with the blocks in the RG transformation, i.e.,
each RG block is a subset of an $L$-block. A subset $X$ of $\L$
defines a subset $\bar{X}$ of $\bar{\L}$, corresponding to the set
of $L$-blocks that have non-empty intersection with $X$.
Conversely, for each site $y$ in $\bar{\L}$, there is a
corresponding $L$-block $y^o$ that is a subset of $\L$. They
divided these $L$-blocks into $2^d$ types, where $d$ is the number
of dimensions of the lattice system. For illustration purposes,
they restricted their attention to two dimensions, so there would
be $4$ types of $L$-blocks, labelled by $i=1,2,3,4$. Let $\sum_i$
denote the summation over the spins which are in a type-$i$
$L$-block. Then trivially,
\begin{equation}
\sum_{\sigma}\exp(-H)=\sum_4 \sum_3 \sum_2 \sum_1 \exp(-H).
\end{equation}

They started by considering $\sum_1 \exp(-H)$ and defined $F^1$ by
$\exp(-F^1)=\sum_1 \exp(-H)$. The sum $\sum_1$ would factor into a
product over type $1$ blocks of the sum over the spins in that
block, $F^1$ is therefore a function of the spins in blocks of
types $2$, $3$, $4$ and the boundary spins. However, when they
tried to compute $\sum_2 \exp(-F^1)$ in a similar fashion, they
ran into difficulty: $F^1$ can contain terms which involve spins
in more than one type $2$ block, so the sum $\sum_2$ does not
factor into a product of independent sums over the type $2$
blocks. To proceed, they distinguished long-range terms
$F_{\text{LR}}^1=\sum_{B: \text{LR}}F_B^1$ supported on sets of
sites with diameter greater than $L$ that prevent the
factorization from short-range terms $F_{\text{SR}}^1=\sum_{B:
\text{SR}}F_B^1$ that do not. Then $\sum_2\exp(-F_{\text{SR}}^1)$
would factor into a product over the type $2$ blocks, and they
defined $F^2$ by $\exp(-F^2)=\sum_2\exp(-F_{\text{SR}}^1)$.

They continued the above constructions iteratively, always
throwing out the long-range terms that prevent the factorization.
Eventually, after performing all the summations, they obtained
$F^4$. For each allowable long-range $B$ (small enough to fit
inside $L$-blocks with side length not exceeding $3L$, thus
consisting of at most $p=3^d$ $L$-blocks), they defined
\begin{eqnarray}
\label{alpha} K(B)=\exp(-F_B^1-F_B^2-F_B^3)-1.
\end{eqnarray}
They then defined a modified expectation $E$, given by
\begin{eqnarray}
Ef=\exp(F^4)\sum_4 \sum_3 \sum_2 \sum_1
\exp(-H+F_{\text{LR}}^1+F_{\text{LR}}^2+F_{\text{LR}}^3)f.
\end{eqnarray}

\section{Cluster expansion}
\label{sec:2} We introduce some combinatorial concepts. A
hypergraph is a set of sites together with a collection $\Gamma$
of nonempty subsets. Such a nonempty set is referred to as a
hyper-edge or link. Two links are $L$-connected if the $L$-blocks
they occupy are within $aL$-distance apart, where $a$ is a
constant that only depends on the number of dimensions $d$ as
shown in \cite{Haller}. A hypergraph $\Gamma$ is $L$-connected if
the support of $\Gamma$ is nonempty and cannot be partitioned into
nonempty sets with no $L$-connected links. We use $\Gamma_C$ to
indicate $L$-connectivity of the hypergraph $\Gamma_C$, and write
$\Gamma_C^*=\cup \bar{\Gamma}_C$ for the support of
$\bar{\Gamma}_C$ in $\bar{\L}$.

As is usual for expansion methods, we work in a finite volume, but
as explained in \cite{Haller}, all estimates are uniform in the
volume and insensitive to boundary conditions, thus the
infinite-volume limit exists according to standard interpretation
of statistical mechanics. In the following most quantities depend
on the finite volume $V$, the choice of boundary condition $\tau$,
and the block spin configuration $\sigma'$, but this dependence is
made implicit.

Haller and Kennedy \cite{Haller} argued that the denominator of
(\ref{part}) has the following cluster representation:
\begin{equation}
\sum_{\sigma}\exp(-H)=e^{-F^4}\sum_{\Delta}\prod_{N\in \Delta}
w_N,
\end{equation}
where $\Delta$ is a set of subsets $N$'s of $\bar{\L}$ (pairwise
at least $a$-distance apart), and
\begin{equation}
\label{w} w_N=\sum_{\Gamma_C^*=N}E\left(\prod_{B\in
\Gamma_C}K(B)\right).
\end{equation}
They justified this by first noticing that it is possible to bound
(\ref{w}) by
\begin{equation}
\label{v} |w_N|\leq v_N=\sum_{\Gamma_C^*=N}\prod_{B\in
\Gamma_C}||K(B)||_{\infty},
\end{equation}
and then showing that under their \textbf{Hypothesis} there is a
function $\epsilon(L)$ such that
\begin{equation}
||K(B)||_{\infty} \leq \epsilon(L)
\end{equation}
for every allowable $B$, with $\epsilon(L)\rightarrow 0$ as
$L\rightarrow \infty$.

We now examine the effect of multiplying $\sigma_W$ to the above
cluster representation as in the numerator of (\ref{part}). There
will be two kinds of terms. In some of these, none of the
$L$-connected components intersect $W$, so for these terms one
gets a product of $\sigma_W$ with a product of independent
$w_N$'s. For the other terms, one decomposes $\Delta$ into a set
of $L$-connected components that is also $L$-connected to $W$ and
remaining ones that are not. We arrive at the representation
\begin{equation}
\sum_{\sigma}\exp(-H)\sigma_W=e^{-F^4}\sum_{R,
\Delta'}\tilde{w}_R\prod_{N\in\Delta'}w_N,
\end{equation}
where $R=\emptyset$ or $R$ and $\bar{W}$ are within $a$-distance
apart, and $\tilde{w}_R$ is a sum over hypergraphs $\Delta_R$ with
$\cup \Delta_R=R$ such that $W$ and $\Delta_R$ are $L$-connected.
($\Delta_R$ itself may not be $L$-connected.) Therefore
\begin{equation}
\label{bound} \frac{\partial J'(Z)}{\partial
J(W)}=\sum_{\sigma'}\sigma'_Z\frac{\sum_{R,
\Delta'}\tilde{w}_R\prod_{N\in\Delta'}w_N}{\sum_{\Delta}\prod_{N\in
\Delta} w_N}.
\end{equation}
We will justify this formal result in the following section.

\section{Existence of the partial derivatives}
\label{sec:3}
\begin{theorem}[Koteck\'{y}-Preiss]
\label{KP} For subsets $N_i$'s of $\bar{\L}$, define
\begin{eqnarray}
\label{c} c(N_1,N_2)=\left\{\begin{array}{ll}
1 & \mbox{if $N_1$ and $N_2$ are within $a$-distance apart};\\
0 & \mbox{otherwise},\end{array} \right.
\end{eqnarray}
and
\begin{equation}
\label{C} C\left(N_1,...,N_n\right)=\sum_{G_c}\prod_{\{i,j\}\in
G_c}\left(-c(N_i,N_j)\right),
\end{equation}
where $G_c$ is a connected graph with vertex set $\{1,...,n\}$.
Take $M>1$. Suppose that for each site $y$ in $\bar{L}$,
\begin{equation}
\label{conv} \sum_{N'}c(N, N')v_{N'}M^{|N'|}\leq |N|\log(M).
\end{equation}
Then the avoidance probability for every $Y\subset \bar{\L}$ has a
convergent power series expansion,
\begin{equation*}
\left|\sum_{\Delta'}\prod_{N\in\Delta'}w_N/\sum_{\Delta}\prod_{N\in
\Delta} w_N\right|
\end{equation*}
\begin{eqnarray*}
=\left|\exp\left(-\sum_{n=1}^{\infty}\frac{1}{n!}\sum_{N_1,...,N_n}C\left(N_1,...,N_n\right)c(Y,
\cup_1^n N_i)w_{N_1}\cdots w_{N_n}\right)\right|
\end{eqnarray*}
\begin{equation}
\label{avoid} \leq
\exp\left(\sum_{n=1}^{\infty}\frac{1}{n!}\sum_{N_1,...,N_n}|C\left(N_1,...,N_n\right)|c(Y,
\cup_1^n N_i)v_{N_1}\cdots v_{N_n}\right)\leq M^{|Y|},
\end{equation}
where $\Delta'$ is a set of subsets of $\bar{\L}$ (pairwise at
least $a$-distance apart) that is also at least $a$-distance away
from $Y$, and $\Delta$ is a set of subsets of $\bar{\L}$ (pairwise
at least $a$-distance apart).
\end{theorem}

\begin{proposition}
\label{M} Take $M>1$. Suppose \textbf{Hypothesis} holds. Suppose
$L$ is sufficiently large so that $\epsilon(L)$ is sufficiently
small,
\begin{equation}
\label{eps} \epsilon(L) \leq \frac{\log(M)(p-1)^p}{rc(Mp)^p
\left(1+(p-1)\log(M)\right)},
\end{equation}
where
\begin{equation}
c=\sum_{m=1}^{p}\sup_{x\in \L}\#\{B: x\in B, |\bar{B}|=m\}<\infty
\end{equation}
due to finite-range and translation-invariant assumptions on the
Hamiltonian, and $r$ is a constant that only depends on the
distance $a$ and the number of dimensions $d$:
\begin{equation}
r=\sup_{y\in \bar{\L}}\#\{z: \text{dist}(y, z)\leq a\}.
\end{equation}
For each site $y$ in $\bar{L}$, let $a_y(N)$ be the collection of
subsets $N\subset \bar{L}$ that satisfy $\text{dist}(y, N)\leq a$.
Then we have
\begin{equation}
\label{2ineq} \sum_{N\in a_y(N)}v_NM^{|N|}\leq \log(M).
\end{equation}
\end{proposition}

\begin{remark}
\textnormal{The inequality (\ref{2ineq}) is a standard sufficient
condition for (\ref{conv}). It will be applied in the following
Theorem \ref{expar}.}
\end{remark}

\begin{proof}
For a fixed but arbitrarily chosen $y$ in $\bar{L}$, we estimate
(\ref{2ineq}).
\begin{eqnarray}
\sum_{N\in a_y(N)}v_N M^{|N|}&=&\sum_{N\in
a_y(N)}\sum_{\Gamma_C^*=N}M^{|N|}\prod_{B\in
\Gamma_C}||K(B)||_{\infty}\\&\leq&\sum_{N\in
a_y(N)}\sum_{\Gamma_C^*=N}M^{p|\Gamma_C|}\left(\epsilon(L)\right)^{|\Gamma_C|}\\&=&\sum_{\Gamma_C:
\text{dist}(y, \Gamma_C^*) \leq
a}\left(M^p\epsilon(L)\right)^{|\Gamma_C|}.
\end{eqnarray}
We say that a hypergraph $\Gamma_C$ is $L$-rooted at $y$ if
$\Gamma_C^*$ and $y$ are within $a$-distance apart. Let $a_n(y)$
be the number of all $L$-connected hypergraphs with $n$ links that
are $L$-rooted at $y$,
\begin{equation}
\label{an} a_n(y)=\#\{\Gamma_C: |\Gamma_C|=n \text{ and dist}(y,
\Gamma_C^*)\leq a\}.
\end{equation}
Let $a_n$ be the supremum over $y$ of the number of $L$-connected
hypergraphs with $n$ links that are $L$-rooted at $y$, i.e.,
$a_n=\sup_{y\in \bar{\L}}a_n(y)$. Then
\begin{equation}
\sum_{N\in a_y(N)}v_N M^{|N|}\leq
\sum_{n=1}^{\infty}a_n\left(M^p\epsilon(L)\right)^n.
\end{equation}

It seems that once we show that $a_n$ grows at most exponentially
with $n$, the geometric series above will converge for small
enough $\epsilon(L)$, and our claim might follow. To estimate
$a_n$, we relate to some standard combinatorial facts
\cite{Minlos}. The rest of the proof follows from a series of
lemmas.
\end{proof}

\begin{lemma}
Let $a_n$ be the supremum over $y$ of the number of $L$-connected
hypergraphs with $n$ links that are $L$-rooted at $y$. Then $a_n$
satisfies the recursive bound
\begin{equation}
a_n\leq rc\sum_{k=0}^{p} \tbinom{p}{k}\sum_{a_{n_1},...,a_{n_k}:
n_1+\cdots+n_k+1=n}a_{n_1}\cdots a_{n_k}
\end{equation}
for $n\geq 1$, where $\tbinom pk$ is the binomial coefficient.
\end{lemma}

\begin{proof}
We first linearly order the points $x$ in $\L$ and also linearly
order the allowable $L$-blocks $B$ of $\L$. This naturally induces
a linear ordering of the points $y$ in $\bar{\L}$. For a fixed but
arbitrarily chosen $y$ in $\bar{\L}$, we examine (\ref{an}). Write
$\Gamma_C=\{B_1\}\cup \Gamma^1_C$, where $B_1$ is the least $B$ in
$\Gamma_C$ with $\text{dist}(y, \bar{B_1})\leq a$. There must be
such an allowable $B_1$, since $\text{dist}(y, \Gamma_C^*)\leq a$.
Moreover, there must be some $z\in \bar{B_1}$ such that
$\text{dist}(y, z)\leq a$, of which there are $r$ possibilities.
Also notice that every $x \in z^o$ will satisfy $x \in B_1$. Thus
\begin{equation}
a_n(y)\leq r\sum_{m=1}^{p}\sup_{x\in \L}\sum_{B_1: x\in B_1,
|\bar{B}_1|=m}\#\{\Gamma^1_C\}.
\end{equation}
As a consequence,
\begin{equation}
a_n(y)\leq rc\#\{\Gamma^1_C\}.
\end{equation}
The remaining hypergraph $\Gamma^1_C$ has $n-1$ subsets and breaks
into $k: k\leq p$ $L$-connected components $\Gamma_1,...,\Gamma_k$
of sizes $n_1,...,n_k$, with $n_1+\cdots+n_k=n-1$. For each
component $\Gamma_i$, there is a least $L$-block $y_i^o$ through
which it is $L$-connected to $B_1$, and the map from the
components $\Gamma_i$ to the $L$-sites $\{y_i\}$ is injective. We
have
\begin{equation}
a_n(y)\leq rc\sum_{k=0}^{p}
\tbinom{p}{k}\sum_{a_{n_1},...,a_{n_k}:
n_1+\cdots+n_k+1=n}a_{n_1}\cdots a_{n_k}.
\end{equation}
Our inductive claim follows by taking the supremum over all $y$ in
$\bar{\L}$. Finally, we look at the base step: $n=1$. In this
simple case, as reasoned above, we have
\begin{eqnarray}
a_1&=&\sup_{y\in \bar{\L}}\#\{\Gamma_C: |\Gamma_C|=1 \text{ and
dist}(y, \Gamma_C^*)\leq
a\}\notag\\&\leq&r\sum_{m=1}^{p}\sup_{x\in \L}\#\{B: x\in B,
|\bar{B}|=m\}\notag\\&=& rc,
\end{eqnarray}
and this verifies our claim.
\end{proof}

Clearly, $\sum_{N\in a_y(N)}v_N M^{|N|}$ will be bounded above by
$\sum_{n=1}^{\infty}\bar{a}_n\left(M^p\epsilon(L)\right)^n$, if
\begin{equation}
\label{a} \bar{a}_n= rc\sum_{k=0}^p\tbinom
pk\sum_{\bar{a}_{n_1},...,\bar{a}_{n_k}:
n_1+\cdots+n_k+1=n}\bar{a}_{n_1}\cdots \bar{a}_{n_k}
\end{equation}
for $n\geq 1$, i.e., equality is obtained in the above lemma.

\begin{lemma}
Consider the coefficients $\bar{a}_n$ that bound the number of
$L$-connected and $L$-rooted hypergraphs with $n$ links. Let
$w=\sum_{n=1}^{\infty}\bar{a}_n z^n$ be the generating function of
these coefficients. Then the recursion relation (\ref{a}) for the
coefficients is equivalent to the formal power series generating
function identity
\begin{equation}
\label{id} w=rcz(1+w)^p.
\end{equation}
\end{lemma}

\begin{proof}
Notice that $(1+w)^p=\sum_{k=0}^p \tbinom pk w^k$, thus
\begin{eqnarray}
w=rcz\sum_{k=0}^p \tbinom pk w^k.
\end{eqnarray}
Writing completely in terms of $z$, we have
\begin{equation}
\sum_{n=1}^{\infty}\bar{a}_n z^n=rc\sum_{k=0}^p \tbinom pk
\sum_{\bar{a}_{n_1},...,\bar{a}_{n_k}:
n_1+\cdots+n_k+1=n}\bar{a}_{n_1}\cdots \bar{a}_{n_k}z^n.
\end{equation}
Our claim follows from term-by-term comparison.
\end{proof}

\begin{lemma}
If $w$ is given as a function of $z$ as a formal power series by
the generating function identity (\ref{id}), then this power
series has a nonzero radius of convergence $|z|\leq
\frac{(p-1)^{p-1}}{rcp^p}$.
\end{lemma}

\begin{proof}
Without loss of generality, assume $z\geq 0$. Set $z_1=rcz$.
Solving (\ref{id}) for $z_1$ gives
\begin{equation}
z_1=\frac{w}{(1+w)^p}.
\end{equation}
By elementary calculus, this increases as $w$ goes from $0$ to
$1/(p-1)$ to have values $z_1$ from $0$ to $(p-1)^{p-1}/{p^p}$. It
follows that as $z_1$ goes from $0$ to $(p-1)^{p-1}/{p^p}$, the
$w$ values range from $0$ to $1/(p-1)$.
\end{proof}

\noindent \textit{Proof of Proposition \ref{M} continued.} We
notice that in the above lemma,
$w=\sum_{n=1}^{\infty}\bar{a}_nz^n=1/(p-1)$ corresponds to
$z_1=rcz=(p-1)^{p-1}/{p^p}$, which implies that for each $n$,
\begin{equation}
\label{imply} \bar{a}_n \leq
\left(rcp^p\right)^n\left(p-1\right)^{-\left(1+(p-1)n\right)}.
\end{equation}
Gathering all the information we have obtained so far,
\begin{eqnarray}
\sum_{N\in a_y(N)}v_N M^{|N|}&\leq&
\sum_{n=1}^{\infty}\left(rc(Mp)^p
\epsilon(L)\right)^n\left(p-1\right)^{-\left(1+(p-1)n\right)}\\&=&\frac{\frac{rc(Mp)^p\epsilon(L)}{(p-1)^p}}{1-\frac{rc(Mp)^p\epsilon(L)}{(p-1)^{p-1}}}\leq
\log(M)
\end{eqnarray}
by (\ref{eps}). \qed

\begin{theorem}
\label{expar} Suppose \textbf{Hypothesis} holds. Then for every
subset $W$ of the original lattice and every subset $Z$ of the
image lattice, the power series expansion of the partial
derivative $\frac{\partial J'(Z)}{\partial J(W)}$ of the RG
transformation (\ref{bound}) converges absolutely.
\end{theorem}

\begin{proof}
The proof of this theorem is an application of the
Koteck\'{y}-Preiss result \cite{Kotecky}. Recall that $N\in
\Delta'$ implies $N$ and $R\cup \bar{W}$ are at least $a$-distance
apart. By the Koteck\'{y}-Preiss theorem (Theorem \ref{KP}),
(\ref{2ineq}) implies
\begin{equation}
\left|\sum_{\Delta'}\prod_{N\in\Delta'}w_N/\sum_{\Delta}\prod_{N\in
\Delta} w_N\right|\leq M^{|R\cup \bar{W}|}.
\end{equation}
To verify our claim, we need to estimate
\begin{equation}
\label{estimate} \sum_{R}|\tilde{w}_R|M^{|R\cup \bar{W}|} \leq
\sum_{\Delta_R}M^{p}\prod_{Y\in \Delta_R}v_Y M^{|Y|}.
\end{equation}
But this is easy, remove $W$, the remaining hypergraph breaks up
into $k: 0\leq k\leq p$ $L$-connected components. So this last
quantity is bounded by
\begin{equation}
\label{suc} M^{p}\sum_{k=0}^{p}\tbinom {p}{k} \left(\log
(M)\right)^k=M^{p} \left(1+\log(M)\right)^{p}.
\end{equation}
\end{proof}

\section{Band structure}
\label{sec:4} By a more complicated application of the cluster
expansion machinery, we show that when \textbf{Hypothesis} holds,
the matrix of partial derivatives displays an approximate band
property.

\begin{proposition}
\label{pin} Suppose \textbf{Hypothesis} holds. Suppose $L$ is
sufficiently large (cf. (\ref{eps})). Then for every site $y$ in
$\bar{L}$, we have
\begin{equation}
\sum_{N\in a_y(N) \text{ and } |N|>P}v_N M^{|N|}\leq \delta(P),
\end{equation}
where \begin{equation} \label{epsilon}
\delta(P)=\frac{\left(\frac{rc(Mp)^p\epsilon(L)}{(p-1)^{p-1}}\right)^{\frac{P}{p}}}{(p-1)\left(1-\frac{rc(Mp)^p\epsilon(L)}{(p-1)^{p-1}}\right)}.
\end{equation}
It is clear that $\delta(P) \rightarrow 0$ as $P \rightarrow
\infty$.
\end{proposition}

\begin{proof}
An $L$-connected hypergraph that is $L$-rooted at $y$ and with
cardinality greater than $P$ will have at least $P/p$ links. This
implies
\begin{equation}
\sum_{N\in a_y(N) \text{ and } |N|>P}v_N M^{|N|}\\
\leq \sum_{n=P/p}^{\infty}\left(rc(Mp)^p
\epsilon(L)\right)^n\left(p-1\right)^{-\left(1+(p-1)n\right)}=\delta(P).
\end{equation}
\end{proof}

\begin{proposition}
\label{pindown} Suppose \textbf{Hypothesis} holds. Suppose $L$ is
sufficiently large (cf. (\ref{eps})). Then for every subset $Y$ of
$\bar{L}$, we have
\begin{equation}
\label{nfinal} \sum_{n=1}^{\infty}\frac{1}{n!}\sum_{N_1,...,N_n:
|\cup_1^n N_i|>P}|C\left(N_1,...,N_n\right)|c(Y, \cup_1^n
N_i)v_{N_1}\cdots v_{N_n} \leq |Y|\delta(P).
\end{equation}
\end{proposition}

\begin{proof}
This follows from Proposition \ref{pin}. Remove $Y$, the remaining
hypergraph is still $L$-connected by (\ref{C}). Moreover, there
can be at most $|Y|$ choices for where it is pinned down.
\end{proof}

\begin{theorem}
\label{band} Suppose \textbf{Hypothesis} holds. Then there is an
approximate band property for the matrix of partial derivatives:
For subset $W$ of the original lattice and subset $Z$ of the image
lattice that are sufficiently far apart, the partial derivative
$\frac{\partial J'(Z)}{\partial J(W)}$ of the RG transformation
(\ref{bound}) is arbitrarily small. Let
\begin{equation}
l(W, Z)=\inf\{\text{dist}(w, z): w\in \bar{W}, z\in \bar{Z}\}
\end{equation}
be the distance between $W$ and $Z$ measured in $\bar{L}$. For
fixed but arbitrary $S, Q$ and $K$, if
\begin{equation} l(W, Z)>(3+a)(pS+QK),
\end{equation}
then
\begin{equation}
\label{enough} |\frac{\partial J'(Z)}{\partial J(W)}| \leq
M^{p}\left(1+\log(M)\right)^{p}\left(\frac{\delta(S)}{\log(M)}+\left(\delta(Q)+\delta(K)\right)p(1+S)M^{p(1+S)}\right).
\end{equation}
\end{theorem}

Before starting the proof of Theorem \ref{band}, let us try to
understand this band property better.

\begin{proposition}
\label{pen} Suppose \textbf{Hypothesis} holds. Then for subset $W$
of the original lattice and subset $Z$ of the image lattice, as
the distance $l(W, Z)$ between $W$ and $Z$ gets large, the partial
derivative $|\frac{\partial J'(Z)}{\partial J(W)}|$ decays
sub-exponentially, a little slower than $\exp(-l(W, Z)^{1/2})$.
\end{proposition}

\begin{proof}
For notational convenience, we denote $l(W, Z)$ simply by $l$.
Take
\begin{equation}
S=\frac{1}{p}\left(\frac{l}{2(3+a)}\right)^{\alpha},
\end{equation}
and
\begin{equation}
Q=K=\left(\frac{l}{2(3+a)}\right)^{\beta},
\end{equation}
where $0<\alpha<\beta\leq 1/2$. We examine (\ref{enough}). The
first factor, $M^{p}\left(1+\log(M)\right)^{p}$, is just a
constant. The second factor is more complicated and thus merits
more attention. The first term, $\delta(S)/\log(M)$, decays as
$\exp\left(-l^{\alpha}\right)$, whereas the second term,
$\left(\delta(Q)+\delta(K)\right)p(1+S)M^{p(1+S)}$, decays as
$\exp\left(-l^{\beta}+l^{\alpha}\right)\sim\exp\left(-l^{\beta}\right)$.
Piecing it all together, $|\frac{\partial J'(Z)}{\partial J(W)}|$
decays sub-exponentially, as $\exp(-l^{\alpha})$.
\end{proof}

\noindent \textit{Proof of Theorem \ref{band}.} Fix an $S$ that is
large enough. We rewrite (\ref{bound}) as
\begin{equation}
\label{small} \frac{\partial J'(Z)}{\partial
J(W)}=\sum_{\sigma'}\sigma'_Z\frac{\sum_{|R|>pS,
\Delta'}\tilde{w}_R\prod_{N\in\Delta'}w_N}{\sum_{\Delta}\prod_{N\in
\Delta} w_N}+\sum_{\sigma'}\sigma'_Z\frac{\sum_{|R|\leq pS,
\Delta'}\tilde{w}_R\prod_{N\in\Delta'}w_N}{\sum_{\Delta}\prod_{N\in
\Delta} w_N}.
\end{equation}
Following, we will verify the smallness of (\ref{small}) by
examining the two terms on the right-hand side separately.

\noindent Case 1: $|R|>pS$. Similarly as in the proof of Theorem
\ref{expar}, we estimate (\ref{estimate}). Remove $W$, the
remaining hypergraph (with cardinality greater than $pS$) breaks
up into $k: 0\leq k\leq p$ $L$-connected components, so at least
one of them has cardinality greater than $S$. By (\ref{suc}) and
(\ref{epsilon}), the contribution of this hypergraph is bounded by
\begin{equation}
M^{p}\delta(S)\sum_{k=0}^{p}\tbinom {p}{k} \left(\log
(M)\right)^{k-1}=\frac{\delta(S)}{\log(M)}M^{p}
\left(1+\log(M)\right)^{p}.
\end{equation}

\noindent Case 2: $|R|\leq pS$. We need to do a more careful
analysis for this case. Recall that $N\in \Delta'$ implies $N$ and
$R\cup \bar{W}$ are at least $a$-distance apart. By the
Koteck\'{y}-Preiss theorem (Theorem \ref{KP}), (\ref{2ineq})
implies
\begin{multline}
\label{exp}
\sum_{\Delta'}\prod_{N\in\Delta'}w_N/\sum_{\Delta}\prod_{N\in
\Delta}
w_N\\=\exp\left(-\sum_{n=1}^{\infty}\frac{1}{n!}\sum_{N_1,...,N_n}C\left(N_1,...,N_n\right)c(R\cup
\bar{W}, \cup_1^n N_i)w_{N_1}\cdots w_{N_n}\right).
\end{multline}
For notational convenience, we will denote the right-hand side of
(\ref{exp}) by $F(\infty, \infty)$, where the first parameter of
$F$ indicates the maximum number of subsets $N_i$'s allowed in the
expansion, and the second parameter of $F$ indicates the
cardinality restriction over these $N_i$'s. It is straightforward
that for fixed $Q$ and $K$,
\begin{equation}
F(\infty, \infty)=F(\infty, \infty)-F(Q, \infty)+F(Q, \infty)-F(Q,
K)+F(Q, K).
\end{equation}

We first examine $F(\infty, \infty)-F(Q, \infty)$. This difference
can be regarded as the tail of the convergent series (\ref{exp}),
thus should be small when $Q$ is large. In fact, it is bounded by
$p(1+S)M^{p(1+S)}\delta(Q)$ by the mean value theorem, applied to
(\ref{avoid}) and (\ref{nfinal}). Fix such a $Q$. We next examine
$F(Q, \infty)-F(Q, K)$. For every subset $N$ of $\bar{L}$, define
\begin{eqnarray} u_N=\left\{\begin{array}{ll}
w_N & \mbox{if $|N|\leq K$};\\
0 & \mbox{otherwise}.\end{array} \right.
\end{eqnarray}
The difference in $F$ can then be interpreted as induced by
evaluating (\ref{exp}) using two sets of parameters $w_N$ and
$u_N$. These two parameter sets both lie in the region of
analyticity of (\ref{exp}), thus intuitively, the difference can
be as small as desired when $K$ is large enough. We again refer to
(\ref{avoid}) and (\ref{nfinal}), and conclude that it is bounded
by $p(1+S)M^{p(1+S)}\delta(K)$. Fix such a $K$. For these two
situations, the only thing left to show now is that
\begin{equation}
\sum_{|R|\leq pS}|\tilde{w}_R|
\end{equation}
is finite, but this naturally follows from (\ref{suc}).

Finally, we examine $F(Q, K)$. As $R\cup \bar{W}$ and $\cup_1^n
N_i$ are within $a$-distance apart, $F(Q, K)$ will only depend on
$L$-sites in a finite region (roughly a ball with radius
$(3+a)(pS+QK)$). If $\bar{Z}$ is outside this region, then
\begin{equation}
\sum_{|R|\leq pS}\tilde{w}_R F(Q, K)
\end{equation}
is a constant with respect to $\sigma'_Z$, thus, when summing over
all possible image configurations $\sigma'$ as in (\ref{small}),
it vanishes. \qed

\section{Upper bound for the RG linearization}
\label{linearization}
\begin{proposition}
\label{ul} Fix a subset $Z$ of the image lattice. Let $n(E)$ be
the number of subsets $W$ of the original lattice that are at most
$E$-distance away from $Z$ (measured in $\bar{L}$),
\begin{equation}
n(E)=\#\{W: l(W, Z)\leq E\}.
\end{equation}
Then $n(E)$ grows polynomially in $E$.
\end{proposition}

\begin{proof}
Due to our finite-range and translation-invariant assumptions on
the Hamiltonian,
\begin{equation}
n=\sup_{y\in \bar{\L}} \#\{W: y\in \bar{W}\}<\infty.
\end{equation}
Thus $n(E)$ grows at the same rate as the volume of a
$d$-dimensional ball with radius $E$, i.e., polynomial growth
$E^d$.
\end{proof}

\begin{theorem}
\label{lin} Suppose \textbf{Hypothesis} holds. Then the
linearization $\mathrm{L}(J_c)$ of the RG transformation
(\ref{linear}) is well-defined and has an upper bound.
\end{theorem}

\begin{proof}
This is mainly due to the fact that sub-exponential decay
dominates polynomial growth. Take $K$ with $||K||$ small (and so,
a fortiori, $||K||_{\infty}$ is small). By Propositions \ref{pen}
and \ref{ul},
\begin{eqnarray}
|\mathrm{L}(J_c)K(Z)|&\leq& \sum_{n=0}^{\infty}\sum_{n\leq l(W, Z)<
n+1} \left|\frac{\partial J'(Z)}{\partial
J(W)}\right|_{J=J_c}\left|K(W)\right|\\ &\lesssim&
||K||_{\infty}\sum_{n=0}^{\infty}\exp(-n^{\alpha})(n+1)^{d}.
\end{eqnarray}
Our claim then follows from the integral test.
\end{proof}

\section*{Acknowledgments}
Part of this work appeared in a PhD dissertation at the University
of Arizona. The author owes deep gratitude to her PhD advisor Bill
Faris for his continued help and support, and to Tom Kennedy, Doug
Pickrell, and Bob Sims for their kind and helpful suggestions and
comments. It is a pleasure to acknowledge stimulating discussions
with participants in the $2011$ renormalization group workshop in
Oberwolfach, organized by Margherita Disertori, Joel Feldman, and
Manfred Salmhofer.

\end{document}